\documentclass[10pt, letterpaper, conference]{ieeeconf}
\IEEEoverridecommandlockouts
\usepackage{amsmath, amsfonts, amssymb, mathtools, url, graphicx, graphicx,transparent}
\usepackage{newtxtext}
\usepackage{newtxmath}
\let\labelindent\relax
\usepackage{enumitem}
\usepackage[usenames, dvipsnames]{color}
\usepackage[colorlinks=true,
        raiselinks=true,
        linkcolor=MidnightBlue,
        urlcolor=PineGreen,
        citecolor=Brown,
        pdfauthor={Atreyee},
        pdftitle={Switched Linear Systems},
        pdfkeywords={continuous-time, t multiply and divide},
        pdfsubject={Technical Report},
        plainpages=false]{hyperref}

\allowdisplaybreaks
\usepackage{tikz}
\usetikzlibrary{automata,positioning}
\usepackage{caption}
\usepackage{subcaption}

\newtheorem{theorem}{Theorem}
\newtheorem{corollary}{Corollary}

\newtheorem{definition}{Definition}
\newtheorem{assumption}{Assumption}

\newtheorem{problem}{Problem}
\newtheorem{remark}{Remark}
\newtheorem{example}{Example}

\newtheorem{algo}{Algorithm}
\newtheorem{proof of theorem}{Proof of Theorem}

\newcommand{\norm}[1]{\left\lVert{#1}\right\rVert}
\newcommand{\abs}[1]{\left\lvert{#1}\right\rvert}

\newcommand{\pmat}[1]{\begin{pmatrix}#1\end{pmatrix}}

\renewcommand{\geq}{\geqslant}

\renewcommand{\leq}{\leqslant}

\newcommand{\R}{\mathbb{R}}
\newcommand{\N}{\mathbb{N}}
\renewcommand{\P}{\mathcal{P}}

\newcommand{\is}{i_s}
\newcommand{\iu}{i_u}
\newcommand{\PP}{\mathbb{P}}

\title{\LARGE \bf
A probabilistic algorithm for scheduling networked control systems under data losses
}

\author{Anubhab Dasgupta, Darsana Udayakumar, Atreyee Kundu
\thanks{Anubhab Dasgupta is with the Department of Mechanical Engineering, Indian Institute of Technology Kharagpur, West Bengal - 721302, India. Darsana Udayakumar and Atreyee Kundu are with the Department of Electrical Engineering, Indian Institute of Technology Kharagpur, West Bengal - 721302, India. Emails: \{anubhab.dasgupta,darsana\_udayakumar\}@kgpian.iitkgp.ac.in, atreyee@ee.iitkgp.ac.in}
}

\begin{document}

\maketitle
\thispagestyle{empty}
\pagestyle{empty}

\begin{abstract}
    This paper deals with the design of scheduling logics for networked control systems (NCSs) whose communication networks have limited capacity and are prone to data losses. Our contributions are twofold. First, we present a probabilistic algorithm to generate a scheduling logic that under certain conditions on the plant and the controller dynamics, the capacity of the network and the probability of data losses, ensures stochastic stability of each plant in the NCS. Second, given the plant dynamics, the capacity of the shared communication network and the probability of data losses, we discuss the design of state-feedback controllers such that our stability conditions are obeyed. Numerical examples are presented to demonstrate the results reported in this paper.
\end{abstract}

\section{Introduction}
\label{s:intro}
        \emph{Networked Control Systems} (NCSs) are spatially distributed control systems in which the communication between plants and their controllers occurs through shared communication networks. NCSs find wide applications in sensor networks, remote surgery, haptics collaboration over the internet, automated highway systems, unmanned aerial vehicles, etc. \cite{Hespanha2007}.
        
         In NCSs the number of plants sharing a communication network is often higher than the capacity of the network. Examples of communication networks with limited bandwidth include wireless networks, which is an important component of smart home, smart transportation, smart city, remote surgery, platoons of autonomous vehicles, etc. and underwater acoustic communication systems \cite{abc}. In addition, the network may undergo communication uncertainties, e.g., data losses. Data loss is a common feature for noisy communication networks. For instance, in cloud-aided vehicle control systems, the control values are computed remotely and transmitted to the vehicles over noisy wireless networks. The interference and fading effects in the noisy network often lead to data losses \cite{Mishra2018}. These scenarios motivate the need to allocate the communication network to each plant in a manner so that good qualitative and quantitative properties of the plants are preserved under communication uncertainties. The task of efficient allocation of a shared communication network is commonly referred to as a \emph{scheduling problem} and the corresponding allocation scheme is called a \emph{scheduling logic}. In this paper we are interested in algorithmic design of scheduling logics for NCSs.     

        Scheduling logics can be classified broadly into two categories: \emph{static} and \emph{dynamic}. In case of the former, a finite length allocation scheme of the network is determined offline and is applied eternally in a periodic manner, while in case of the latter, the allocation of the shared network is determined based on some information about the plant (e.g., states, outputs, access status of sensors and actuators, etc.).  
        
        Static scheduling logics are designed using common Lyapunov functions in \cite{Hristu2001}, piecewise Lyapunov-like functions with average dwell time switching in \cite{Lin2005}, combinatorial optimization with periodic control theory in \cite{Rehbinder2004}, Linear Matrix Inequalities (LMIs) optimization with average dwell time technique in \cite{Dai2010}, multiple Lyapunov-like functions and graph theory in \cite{def}, and stochastic Lyapunov functions in \cite{ghi}. Dynamic scheduling logics are designed using event-triggering in \cite{Al-Areqi'15}, optimal node minimization in \cite{Quevedo2014}, distributed control-aware random network access logics in \cite{Gatsis2016} and linear programming techniques in \cite{Ma2019}.
        
       Recently in \cite{abc} the authors consider a shift in paradigm and propose a probabilistic algorithm to design scheduling logics. At every instant of time, this algorithm allocates the shared network to subsets of the plants with certain probabilities. The authors present necessary and sufficient conditions on the plant dynamics and the capacity of the shared network under which a scheduling logic obtained from this algorithm ensures stochastic stability of each plant in the NCS. Given the plant dynamics and capacity of the shared network, they also present an algorithm to design static state-feedback controllers such that the plants, their controllers and the shared network together satisfy the proposed stability conditions. The stability results of \cite{abc}, however, suffer from two classes of limitations:
       \begin{itemize}[label=\(\circ\),leftmargin=*]
            \item The communication network has a limited bandwidth but does not suffer from communication uncertainties such as delays, data losses, etc. In practice, ideal communication is rare.
            \item The capacity of the communication network is such that the total number of plants is divisible by the communication capacity. This restricts the class of NCSs under consideration.
       \end{itemize}
       In this paper we report the applicability of the probabilistic scheduling algorithm presented in \cite{abc} to a broader premise in NCS literature. In particular, we cater to
       \begin{itemize}[label=\(\circ\),leftmargin=*]
            \item NCSs whose shared communication networks have a limited capacity and are prone to data losses. In particular, we consider a probabilistic data loss model.
            \item NCSs with no restriction on the number of plants and the capacity of the communication network. In particular, the total number of plants is not necessarily divisible by the capacity of the communication network.
       \end{itemize}
       
       We work on the following premise: We consider an NCS consisting of multiple discrete-time linear plants whose feedback loops are closed through a shared communication network. We assume that the plants are unstable in open-loop and the controllers are stabilizing. Due to a limited communication capacity of the network, only a few plants can exchange information with their controllers at any instant of time. Consequently, the remaining plants operate in open-loop at every time instant. Additionally, the plants communicating with their controllers at an instant of time may operate in open-loop as the control input may be lost in the network. Our specific contributions are the following: 
       \begin{itemize}[label = \(\circ\), leftmargin = *]
            \item First, we consider the probabilistic scheduling algorithm of \cite{abc} and provide necessary and sufficient conditions on the plant and the controller dynamics, the capacity of the network and the probability of data losses, such that a scheduling logic obtained from this algorithm ensures stochastic stability of each plant in the NCS under data losses. 
            \item Second, given the plant dynamics, the capacity of the shared communication network and the probability of data losses, we discuss the design of state-feedback controllers such that our stability conditions are satisfied. 
       \end{itemize}
        A generalized version of the stability condition presented in \cite{abc} falls as a special case of our results. Similar to \cite{abc}, the key tool of our analysis is a Markovian jump linear system modelling of the plants in the NCS. Numerical examples are presented to demonstrate our results.
       
       The remainder of this paper is organized as follows: In Section \ref{s:prob_stat} we formulate the problem under
consideration. Our results are presented in Section \ref{s:mainres}. Numerical experiments are discussed in Section \ref{s:numex}. We conclude in Section \ref{s:concln} with a brief discussion on future research direction.

        {\bf Notation}. \(\R\) is the set of real numbers and \(\N\) is the set of natural numbers, \(\N_0 = \N\cup\{0\}\). For two scalars \(a\) and \(b\), \(a\%b\) denotes the remainder of the operation \(a/b\). For a finite set \(C\), its cardinality is denoted by \(\abs{C}\). For a vector \(v\), \(\norm{v}\) denotes its Euclidean norm. \(0_{d}\) denotes \(d\)-dimensional \(0\)-vector and \(0_{d\times d}\) denotes \(d\)-dimensional \(0\)-matrix. We will operate in a probabilistic space \((\Omega,\mathcal{F},\PP)\), where \(\Omega\) is the sample space, \(\mathcal{F}\) is the \(\sigma\)-algebra of events, and \(\PP\) is the probability measure.
\section{Problem statement}
\label{s:prob_stat}
    Consider an NCS with $ N $ plants whose dynamics are given by
\begin{align}
		\label{e:plants}
	x_{i}(t+1)=A_{i} x_{i}(t) + B_{i} u_{i}(t),\:x_{i}(0)= x_{i}^0,\:t\in\mathbb{N}_{0} ,
\end{align} 
where $ x_{i}(t) \in\mathbb{R}^{d_{i}} $ and $ u_{i}(t) \in \mathbb{R}^{m_{i}} $ are the vectors of states and inputs of the $ i $-th plant at time $ t $, respectively, $ i=1,2, \ldots, N $. Each plant \(i\) employs a remotely located state-feedback controller $ u_{i}(t)= K_{i} x_{i}(t),\ t\in\mathbb{N}_{0} $. We have that $ A_i\in\mathbb{R}^{{d_{i}}\times {d_{i}}} $, $ B_i\in\mathbb{R}^{{d_{i}}\times {m_{i}}} $ and $ K_i\in\mathbb{R}^{{m_{i}}\times {d_{i}}} $, $ i=1,2, \ldots, N $, are constants.

The plants communicate with their controllers through a shared communication network. The network has a limited communication capacity in the following sense: at any time instant, at most $ M $ plants $ (0 < M < N) $ can access the network. Consequently, the remaining plants operate in open-loop, i.e., with $ u_{i}(t)=0_{m_{i}} $. In addition, the shared network suffers from communication uncertainty of the following form: the communication channels from the controllers to the plants are prone to data losses. In particular, at any time instant, the control input is lost in a channel with a probability $ q $. The plant accessing this channel at that time instant also operates in open-loop, i.e., with $ u_{i}(t)=0_{m_{i}} $.

Notice that in our setting each plant operates in two modes:

(a) closed-loop operation (i.e., with $ u_{i}(t)= K_{i} x_{i}(t) $) when the plant has access to the shared communication network and its control input is received in good order, and

(b) open-loop operation (i.e., with $ u_{i}(t)=0_{m_{i}} $) when the plant does not have access to the shared communication network, or it has access to the shared network but its control input is lost in the channel. 
\begin{assumption}
	\label{a:stab_unstab}
	\rm{
		For each plant in the NCS, the open-loop dynamics is unstable and the controller gain is stabilizing. More specifically, the matrices $ A_{i},\ i =1,2, \ldots, N  $ are unstable and the matrices $ A_{i}+ B_{i} K_{i},\ i =1,2,\ldots, N $ are Schur.\footnote{{A matrix $ A\in \mathbb R^{d\times d} $ is Schur if all its eigenvalues are inside the open unit disk and is non-Schur otherwise.}
}
}
\end{assumption}

It follows that the closed-loop operation of each plant is stable while the open-loop operation of each plant is unstable. We let $ i_{s} $ and $ i_{u} $ denote the stable and unstable mode of operation of plant $ i $ respectively, $A_{i_{s}}=A_{i}+ B_{i} K_{i} $, $ A_{i_{u}}=A_{i}, i=1,2, \ldots, N $.

Let $ \kappa_{m} : \mathbb{N}_{0} \rightarrow \{0,1\} $ denote the data loss signal at the $ m $-th channel of the communication network, $ m=1,2, \ldots, M $. If $ \kappa_{m}(t)=0 $, then there is no data loss in the $m$-th channel at time $ t $, and if $ \kappa_{m}(t)=1 $, then there is a data loss in the $ m $-th channel at time $ t $. We call a data loss signal $ \textit{admissible} $ if it satisfies
\begin{align*}
 \kappa_{m}(t)= 
\begin{cases}
1, \:&\text{with probability}\  q,\\ 
0, \:&\text{with probability}\  1-q,\\
\end{cases} 
t\in\mathbb{N}_{0}, 
\end{align*}
\(m=1,2,\ldots, M\).
Let $ \mathcal S $ be the set of all distinct subsets of $ \{1,2, \ldots, N \} $ with at most $ M $ elements. We call a function $ \gamma : \mathbb{N}_{0} \rightarrow \mathcal S $, that specifies, at every time $ t $, at most $ M $ plants of the NCS which has access to the shared communication network at that time, as a $ \textit{scheduling logic} $.

We let $ r_{i}^{0} $ denote the initial mode of operation of plant $i$,\\ i.e.,
$ r_{i}^{0} = i_{s} $ , if $ i \in \gamma(0) $ and the communication channel accessed by plant $ i $ does not suffer from data loss at time $ 0 $, and
$ r_{i}^{0} = i_{u} $ , if $ i \notin \gamma(0) $ or, $ i \in \gamma(0) $ but there is a data loss in the communication channel accessed by the plant $ i $ at time $ 0 $.

We are interested in stochastic stability of all plants in the NCS.

\begin{definition} 
	\label{d:stoch_stab}
	\rm{
		The $ i $-th plant in \eqref{e:plants} is \textit{stochastically stable} if for every initial condition $ x_{i}^{0} \in \mathbb{R}^{d_{i}} $ and initial mode of operation $ r_{i}^{0} \in \{i_{s},i_{u}\} $, we have that the following condition holds :
		\begin{equation}\label{def1}
			\mathbb{E}\left\{\sum\limits_{t=0}^{+ \infty} \lVert {x_{i}(t)} \rVert ^2 \ | \ x_{i}^0, r_{i}^0\right\} < +\infty. 
		\end{equation}
	}
\end{definition}

    We will first solve the following problem :
\begin{problem}
\label{prob:main1}
\rm{
Design a scheduling logic, $ \gamma $, that preserves stochastic stability of each plant $ i= 1,2,\ldots, N $ under admissible data loss signals, $ \kappa_{m}, m=1,2, \ldots, M $.
}
\end{problem}	

    Towards solving Problem \ref{prob:main1}, we rely on the probabilistic scheduling algorithm proposed in \cite{abc} and provide sufficient conditions on the matrices $ A_{i_{s}}, A_{i_{u}}, i=1,2, \ldots, N $, the capacity of the communication network, $ M $, and the probability of data loss, $ q $, such that a scheduling logic, \(\gamma\), obtained from this algorithm ensures stochastic stability of each plant in the NCS. 

    We will then focus on the design of stabilizing controllers such that our stability conditions are satisfied. In particular, we are interested in the following problem:
\begin{problem}
	\label{prob:main2}
	\rm{
Given the plant dynamics, $(A_{i},B_{i})$, $ i=1,2, \ldots, N $, the capacity 
of the communication network, $ M $, and the probability of data loss, $ q $, design state-feedback controllers, $ K_{i}, i=1,2, \ldots, N $, such that the conditions for stability under our scheduling logic are satisfied.
	}
\end{problem}
    Towards solving Problem \ref{prob:main2}, we rely on matrix inequalities involving the matrices, \((A_i,B_i)\), \(i=1,2,\ldots,N\), and certain other quantities involved in our scheduling algorithm.
    
    Finally, we will identify a recent result on scheduling under limited bandwidth but ideal communication \cite{abc} as a special case of the results presented in this paper.
    
    We now present our results.

\section{Main results}
\label{s:mainres}
    We begin with our solution to Problem \ref{prob:main1}. 
    
   Consider sets \(s_j\subseteq\{1,2,\ldots,N\}\), \(j=1,2,\ldots,v\) and scalars \(p_j\in]0,1[\), \(j=1,2,\ldots,v\) such that the following conditions are satisfied:\footnote{An estimate of \(v\) will be discussed momentarily.}
        \begin{enumerate}[label= C\arabic*), leftmargin = *]
            \item\label{c1} \(\abs{s_j}\leq M\) for all \(j=1,2,\ldots,v\),
            \item\label{c2} \(s_j\cap s_k = \emptyset\) for all \(j,k=1,2,\ldots,v\), \(j\neq k\),
            \item\label{c3} \(\displaystyle{\bigcup_{j=1}^{v}s_j = \{1,2,\ldots,N\}}\), and
            \item\label{c4} \(\displaystyle{\sum_{j=1}^{v}p_j = 1}\).
        \end{enumerate}
    
    The following probabilistic algorithm employs \(s_j\subseteq\{1,2,\ldots,N\}\), \(j=1,2,\ldots,v\) and \(p_j\in]0,1[\), \(j=1,2,\ldots,v\) to generate scheduling logics, \(\gamma\). 
    \begin{algo}
        \label{algo:sched_design}
            {\it Design of scheduling logic, \(\gamma\)}\\
            \hspace*{1cm}For \(t=0,1,2,\ldots\)\\
            \hspace*{1.5cm} Set \(\gamma(t)=s_j\) with probability \(p_j\).\\
            \hspace*{1cm}End For
        \end{algo}
        
        Algorithm \ref{t:mainres1} relies on a set of sets, \(s_j\subseteq\{1,2,\ldots,N\}\), \(j=1,2,\ldots,v\) and a set of scalars, \(p_j\in]0,1[\), \(j=1,2,\ldots,v\) that satisfy conditions \ref{c1}-\ref{c4}. Conditions \ref{c1}-\ref{c3} govern properties of the sets \(s_j\), \(j=1,2,\ldots,v\). In particular, condition \ref{c1} ensures that each \(s_j\), \(j=1,2,\ldots,v\) has at most \(M\) elements, condition \ref{c2} ensures that the sets \(s_j\), \(j=1,2,\ldots,v\) are pairwise disjoint, and condition \ref{c3} ensures that the sets \(s_j\), \(j=1,2,\ldots,v\) together contain all elements of the set \(\{1,2,\ldots,N\}\). Clearly, we need \(\bigl\lceil\frac{N}{M}\bigr\rceil\leq v\leq N\). Condition \ref{c4} governs properties of the scalars \(p_j\in]0,1[\), \(j=1,2,\ldots,v\). In particular, it ensures that these scalars sum up to \(1\).  
        
        We employ the sets \(s_j\), \(j=1,2,\ldots,v\) as the subsets of the set of all plants that are to be allocated access to the shared communication network. It is immediate that \(s_j\), \(j=1,2,\ldots,v\) are elements of the set \(\mathcal{S}\). Conditions \ref{c1}-\ref{c3} together ensure that each \(s_j\), \(j\in\{1,2,\ldots,v\}\) has at most \(M\)-many plants and each plant \(i\in\{1,2,\ldots,N\}\) appears in exactly one \(s_j\), \(j\in\{1,2,\ldots,v\}\). The size of \(s_j\), \(j=1,2,\ldots,v\) is governed by the capacity of the communication network. Selection of each plant in at least one \(s_j\), \(j\in\{1,2,\ldots,v\}\) is necessary since by Assumption \ref{a:stab_unstab}, each plant is unstable in open-loop and requires access to the shared network for communicating with its stabilizing controller. Selection of each plant in exactly one \(s_j\), \(j\in\{1,2,\ldots,v\}\) is specific to our solution setting and will be employed in the proof of our stability result. We employ the scalars \(p_j\), \(j=1,2,\ldots,v\) as the probability of selecting a set \(s_j\), \(j=1,2,\ldots,v\) to be assigned to \(\gamma\) at any time instant. More specifically, consider that the sets \(s_j\), \(j=1,2,\ldots,v\) and the scalars \(p_j\), \(j=1,2,\ldots,v\) are fixed. In Algorithm \ref{algo:sched_design}, we design a scheduling logic, \(\gamma\), as follows: at each time instant \(t=0,1,2,\ldots\), we allocate the shared communication network to the plants in \(s_j\), \(j\in\{1,2,\ldots,v\}\) with a probability \(p_j\), \(j\in\{1,2,\ldots,v\}\).
        
        \begin{remark}
        \label{rem:compa1}
            Earlier in \cite{abc} the notion of probabilistic allocation of the shared communication network to various subsets of the set of all plants was used. Algorithm \ref{algo:sched_design} differs from \cite[Algorithm 1]{abc} in the choice of subsets of the set of all plants. In \cite{abc} the authors consider that the capacity of the communication network, \(M\), divides the total number of plants, \(N\), i.e., \(N\% M = 0\) and each \(s_j\) contains exactly \(M\) plants, \(j=1,2,\ldots,\frac{N}{M}\). Our choice of \(s_j\), \(j=1,2,\ldots,v\) does not require \(N\% M = 0\) and thus caters to a larger class of NCSs. A more important difference between Algorithm \ref{algo:sched_design} and \cite[Algorithm 1]{abc} lies in the premise of scheduling. While in \cite{abc} analysis of a scheduling logic, \(\gamma\), obtained from a probabilistic algorithm is performed for NCSs whose shared communication networks have a limited bandwidth but ideal communication, in the current paper we deal with limited bandwidth \emph{and} data losses. We will identify a generalized version of the stability conditions of \cite{abc} as a special case of the results presented here.
        \end{remark}
        
         We now provide sufficient conditions on the plant dynamics, the capacity of the communication network and the probability of data loss under which a scheduling logic, \(\gamma\), obtained from Algorithm \ref{algo:sched_design} ensures stochastic stability of each plant in the NCS. 
     \begin{theorem}
    \label{t:mainres1}
        Consider an NCS described in Section \ref{s:prob_stat}. Let the plant dynamics, \((A_i,B_i)\), \(i=1,2,\ldots,N\), the controller dynamics, \(K_i\), \(i=1,2,\ldots,N\), the capacity of the communication network, \(M\), and the probability of data loss, \(p\), be given. Each plant \(i=1,2,\ldots,N\) is stochastically stable under a scheduling logic, \(\gamma\), obtained from Algorithm \ref{algo:sched_design} if and only if for each \(i\in s_j\), \(j=1,2,\ldots,v\), there exist symmetric and positive definite matrices, \(P_k\in\R^{d_i\times d_i}\), \(k=i_s,i_u\), such that 
                \begin{align}
                \label{e:maincondn1}
                    A_k^\top \P^{i}A_k - P_k \prec 0_{d_i\times d_i},
                \end{align}    
                where \(\P^i = p_j(1-q)P_{i_s}+\bigl((1-p_j)+p_jq\bigr)P_{i_u}\).
    \end{theorem}

    Algorithm \ref{algo:sched_design} and Theorem \ref{t:mainres1} together form our solution to Problem \ref{prob:main1}. Condition \eqref{e:maincondn1} involves properties of both the sets \(s_j\), \(j=1,2,\ldots,v\) and the scalars \(p_j\), \(j=1,2,\ldots,v\), and the data loss probability, \(q\). It ensures that the dynamics corresponding to the indices picked in \(s_j\), \(j=1,2,\ldots,v\), the scalars \(p_j\), \(j=1,2,\ldots,v\) and the scalar \(q\), together satisfy an inequality involving the existence of certain positive definite matrices, \(P_{i_s}, P_{i_u}\), \(i\in s_j\), \(j=1,2,\ldots,v\). The choice of \(s_j\), \(j=1,2,\ldots,v\) depends on the capacity of the communication network. Thus, our sufficient condition for stability under \(\gamma\) obtained from Algorithm \ref{algo:sched_design}, involves the plant dynamics, the capacity of the communication network and the probability of data loss.
        
        \begin{remark}
        \label{rem:compa2}
            In \cite[Theorem 1]{abc} necessary and sufficient conditions for stochastic stability of each plant \(i=1,2,\ldots,N\) under a scheduling logic, \(\gamma\), obtained from a probabilistic algorithm was presented. The said result is specific to the setting where the shared communication network has a limited communication capacity but does not suffer from communication uncertainties like delays, data losses, etc. Theorem \ref{t:mainres1} in the current paper provides stability condition under probabilistic scheduling when the shared communication network has a limited capacity and is prone to data losses.
        \end{remark}
        
        \begin{remark}
        \label{rem:compa3}
             The mathematical apparatus that allows us to arrive at \cite[Theorem 1]{abc} and Theorem \ref{t:mainres1} in the current paper is a Markovian jump linear system representation of the individual plants in an NCS. The stability conditions of \cite{abc} and the current paper differ in the definition of \(\P^i\) --- it does not include the data loss probability, \(q\), in \cite{abc}. This is obtained from the difference between the Markovian jump linear systems model of the plant dynamics under limited bandwidth \emph{but} no data loss and the model of the plant dynamics under limited bandwidth \emph{and} data losses. In both the cases off-the-shelf stability conditions for Markovian jump linear systems suffice for the analysis. 
        \end{remark}
        
        \begin{proof}[Proof of Theorem \ref{t:mainres1} (Sketch)]
            Fix a scheduling logic, \(\gamma\), obtained from Algorithm \ref{algo:sched_design}. Fix \(j\in\{1,2,\ldots,v\}\) and \(i\in s_j\). By the properties of \(s_j\), \(j=1,2,\ldots,v\), the plant \(i\) does not appear in any \(s_k\), \(k\in\{1,2,\ldots,v\}\setminus\{j\}\). The plant \(i\) under \(\gamma\) can be modelled as a switched linear system
            \begin{align*}
                x_i(t+1) = A_{\sigma_i(t)}x_i(t),\:\sigma_i(t)\in\{\is,\iu\}
            \end{align*}
            whose set of subsystems is \(\{\is,\iu\}\) and the transition function, \(\sigma_i:\N_0\to\{\is,\iu\}\), satisfies \(\sigma_i(t) = \is\), if \(i\in\gamma(t)\) and \(\kappa_m = 0\) where \(i\) is the \(m\)-th element of \(\gamma\) and \(\sigma_i(t) = \iu\), if \(i\notin\gamma(t)\) or \(i\in\gamma(t)\) and \(\kappa_m = 1\) where \(i\) is the \(m\)-th element of \(\gamma\), \(m\in\{1,2,\ldots,M\}\). Clearly, \(\sigma_i\) is a Markov chain, defined on \((\Omega,\mathcal{F},\PP)\), taking values in \(\{\is,\iu\}\) with transition probability matrix
            \begin{align*}
                \Pi_i = \pmat{\pi_{\is\is} & \pi_{\is\iu}\\\pi_{\iu\is} & \pi_{\iu\iu}}, 
            \end{align*}
            where
            \begin{align*}
            \begin{aligned}
                \pi_{\is\is} &= \PP\bigl(\sigma_i(t+1)=\is\:|\:\sigma_i(t)=\is\bigr) = p_j(1-q),\\
                \pi_{\is\iu} &= \PP\bigl(\sigma_i(t+1)=\iu\:|\:\sigma_i(t)=\is\bigr) = (1-p_j)+p_j q,\\
                \pi_{\iu\is} &= \PP\bigl(\sigma_i(t+1)=\is\:|\:\sigma_i(t)=\iu\bigr) = p_j(1-q),\\
                \pi_{\iu\iu} &= \PP\bigl(\sigma_i(t+1)=\iu\:|\:\sigma_i(t)=\iu\bigr) = (1-p_j)+p_j q,
            \end{aligned}
            \end{align*}
            \(i\in s_j\). The assertion of Theorem \ref{t:mainres1} follows under the set of arguments employed in the proof of \cite[Theorem 1]{abc}.
        \end{proof}
        
        \begin{remark}
        \label{rem:compa4}
            Our design of a scheduling logic is through a probabilistic algorithm. A scheduling logic obtained from Algorithm \ref{algo:sched_design} is neither \emph{static} nor \emph{dynamic} (The reader is referred to Section \ref{s:intro} for a description of these terms). Indeed, we do not repeat a finite length allocation scheme or take properties of the plants or other components in the NCS into consideration at every time instant. Our design is, however, close in spirit to static scheduling techniques in the following sense: (a) computation of the sets, \(s_j\), \(j=1,2,\ldots,v\) and the probabilities, \(p_j\), \(j=1,2,\ldots,v\) are performed offline prior to execution of the scheduling algorithm, and (b) a logic, \(\gamma\), obtained from Algorithm \ref{algo:sched_design} does not adapt to unforeseen changes/faults in the plants or other components in the NCS.
        \end{remark}
        
        In the absence of data losses, design of a scheduling logic, \(\gamma\), that ensures stochastic stability of each plant in an NCS, follows as a special case of Theorem \ref{t:mainres1}.
        
        \begin{corollary}
        \label{cor:mainres2}
            Consider an NCS described in Section \ref{s:prob_stat}. Let the plant dynamics, \((A_i,B_i)\), \(i=1,2,\ldots,N\), the controller dynamics, \(K_i\), \(i=1,2,\ldots,N\), and the capacity of the communication network, \(M\), be given. Let the probability of data loss, \(p=0\). Each plant \(i=1,2,\ldots,N\) is stochastically stable under a scheduling logic, \(\gamma\), obtained from Algorithm \ref{algo:sched_design} if and only if for each \(i\in s_j\), \(j=1,2,\ldots,v\), there exist symmetric and positive definite matrices, \(P_k\in\R^{d_i\times d_i}\), \(k=i_s,i_u\), such that 
                \begin{align}
                \label{e:maincondn2}
                    A_k^\top \tilde{\P}^{i}A_k - P_k \prec 0_{d_i\times d_i},
                \end{align}    
                where \(\tilde{\P}^i = p_j P_{i_s}+(1-p_j)P_{i_u}\).
        \end{corollary}
        
        \begin{proof}
            Since \(q=0\), 
            \begin{align*}
                \tilde\P^i &= p_j P_{\is}+(1-p_j) P_{\iu}\\
                &= p_j (1-0) P_{\is} + \bigl((1-p_j)+0\bigr) P_{\iu}\\
                &= p_j(1-q)P_{\is} + \bigl((1-p_j)+p_j q\bigr) P_{\iu}.
            \end{align*}
            The assertion of Corollary \ref{cor:mainres2} follows directly from Theorem \ref{t:mainres1}.
        \end{proof}
          
        \begin{remark}
        \label{rem:compa5}
            Notice that Corollary \ref{cor:mainres2} is more general than \cite[Theorem 1]{abc}. Indeed, Corollary \ref{cor:mainres2} does not require the capacity of the communication network, \(M\), to divide the number of plants, \(N\), (i.e., \(N\% M=0\)) and hence caters to a larger class of NCSs. 
        \end{remark}
        
        Notice that while we do not have control over the plant dynamics, \((A_i,B_i)\), \(i=1,2,\ldots,N\), the capacity of the communication network, \(M\), and the probability of data loss, \(q\), there is an element of choice associated to the sets, \(s_j\), \(j=1,2,\ldots,v\), the probabilities, \(p_j\), \(j=1,2,\ldots,v\) and the controller dynamics, \(K_i\), \(i=1,2,\ldots,N\). Indeed, given \((A_i,B_i)\), \(i=1,2,\ldots,N\), \(M\) and \(q\), one would ideally like to ``co-design'' \(s_j\), \(j=1,2,\ldots,v\), \(p_j\), \(j=1,2,\ldots,v\) and \(K_i\), \(i=1,2,\ldots,N\) such that conditions \ref{c1}-\ref{c4} and \eqref{e:maincondn1} hold. However, this co-design is a numerically difficult problem. A way to solve it is by performing an exhaustive search for \(K_i\), \(i=1,2,\ldots,N\) over all choices of \(s_j\) $ \subseteq \{1,2, \ldots, N\} $, \(j=1,2,\ldots,v\) satisfying conditions \ref{c1}-\ref{c3} and \(p_j\in]0,1[\), \(j=1,2,\ldots,v\) satisfying condition \ref{c4} until a satisfactory combination is obtained. This motivates our Problem \ref{prob:main2}. 
        
        Given the sets, \(s_j\), \(j=1,2,\ldots,v\) and the probabilities, \(p_j\), \(j=1,2,\ldots,v\), the following result allows us to design state-feedback controllers, \(K_i\), \(i=1,2,\ldots,N\).
        \begin{theorem}
        \label{t:mainres4}
            Consider an NCS described in Section \ref{s:prob_stat}. Let the plant dynamics, \((A_i,B_i)\), \(i=1,2,\ldots,N\), the sets, \(s_j\), \(j=1,2,\ldots,v\) and the probabilities, \(p_j\), \(j=1,2,\ldots,v\) be given. Suppose that the state-feedback controllers, \(K_i\), \(i=1,2,\ldots,N\), are computed as 
            \begin{align}
            \label{e:maincondn4}
                K_i = Y_i P_{i_s},\:\:i=1,2,\ldots,N,
            \end{align}
            where \(Y_i\in\R^{m_i\times d_i}\) is a solution to 
            \begin{align}
            \label{e:auxcondn1}
                \bigl(A_i P_{i_s}^{-1}+B_i Y_i\bigr)^\top (\P_i)^{-1} \bigl(A_i P_{i_s}^{-1}+B_i Y_i\bigr) - P_{i_s}^{-1} \prec 0_{d_i\times d_i},
            \end{align}
            where \(P_{i_s}\), \(P_{i_u}\in\R^{d_i\times d_i}\) are symmetric and positive definite matrices satisfying 
            \begin{align}
            \label{e:auxcondn2}
                A_{i_u}^\top \P^i A_{i_u} - P_{i_u} \prec 0_{d_i\times d_i},
            \end{align}
            and \(\P^i = p_j(1-q)P_{i_s}+\bigl((1-p_j)+p_j q\bigr)P_{i_u}\). Then condition \eqref{e:maincondn1} holds.
        \end{theorem}
        
        In Theorem \ref{t:mainres4} the state-feedback controllers, \(K_i\), \(i=1,2,\ldots,N\) are designed by involving the matrices, \(Y_i\) and \(P_{i_s}\), \(i=1,2,\ldots,N\). These matrices are solutions to the inequalities \eqref{e:auxcondn1} and \eqref{e:auxcondn2}, respectively both of which involve the plant dynamics. Notice that we consider the sets \(s_j\), \(j\in\{1,2,\ldots,v\}\) and their corresponding probabilities \(p_j\), \(j\in\{1,2,\ldots,v\}\) to be fixed a priori. Which \(s_j\), \(j\in\{1,2,\ldots,v\}\) a particular plant \(i\in\{1,2,\ldots,N\}\) belongs to and with what probability \(p_j\), \(j\in\{1,2,\ldots,v\}\) this plant is selected, plays a role in conditions \eqref{e:auxcondn1} and \eqref{e:auxcondn2}. Indeed, the matrix \(\P^i\) involves the probabilities \(p_j\), \(j=1,2,\ldots,v\). 
        
        The design of state-feedback controllers, \(K_i\), \(i=1,2,\ldots,N\) presented in \cite[Theorem 2]{abc} is a special case of Theorem \ref{t:mainres4} above when the probability of data loss, \(q=0\). Notice that since in the design of state-feedback controllers, the sets \(s_j\), \(j=1,2,\ldots,v\) and the probabilities \(p_j\), \(j=1,2,\ldots,v\) are considered to be fixed a priori, the \(N\% M = 0\) dependent choice of \(s_j\) does not lose generality unlike Theorem \ref{t:mainres1} and falls directly as a special case.
        
        \begin{corollary}
        \label{cor:mainres5}
            Consider an NCS described in Section \ref{s:prob_stat}. Let the probability of data loss, \(q=0\). Let the plant dynamics, \((A_i,B_i)\), \(i=1,2,\ldots,N\), the sets, \(s_j\), \(j=1,2,\ldots,v\) and the probabilities, \(p_j\), \(j=1,2,\ldots,v\) be given. Suppose that the state-feedback controllers, \(K_i\), \(i=1,2,\ldots,N\), are computed as 
            \begin{align}
            \label{e:maincondn5}
                K_i = Y_i P_{i_s},\:\:i=1,2,\ldots,N,
            \end{align}
            where \(Y_i\in\R^{m_i\times d_i}\) is a solution to 
            \begin{align}
            \label{e:auxcondn3}
                \bigl(A_i P_{i_s}^{-1}+B_i Y_i\bigr)^\top (\tilde{\P}_i)^{-1} \bigl(A_i P_{i_s}^{-1}+B_i Y_i\bigr) - P_{i_s}^{-1} \prec 0_{d_i\times d_i},
            \end{align}
            \(P_{i_s}\), \(P_{i_u}\in\R^{d_i\times d_i}\) are symmetric and positive definite matrices satisfying 
            \begin{align}
            \label{e:auxcondn4}
                A_{i_u}^\top \tilde{\P}^i A_{i_u} - P_{i_u} \prec 0_{d_i\times d_i},
            \end{align}
            and \(\tilde{\P}^i = p_jP_{i_s}+(1-p_j)P_{i_u}\). Then condition \eqref{e:maincondn1} holds.
        \end{corollary}
        
       \begin{proof}
            Follows immediately from the argument employed in our proof of Corollary \ref{cor:mainres2}.
       \end{proof}
        
        \begin{remark}
        \label{rem:compa6}
            The design procedure of state-feedback controllers, \(K_i\), \(i=1,2,\ldots,N\), in Theorem \ref{t:mainres4} follows from off-the-shelf design of state-feedback controllers for Markovian jump linear systems. In particular, it is an extension of off-the-shelf results to the setting of stability of \(N\) plants simultaneously. We skip the proof of Theorem \ref{t:mainres4} as it follows directly from the proof of \cite[Theorem 2]{abc} by modifying the definition of \(\P^i\).
        \end{remark}
        
        \begin{remark}
        \label{rem:solve}
            Notice that the sets of matrix inequalities involved in Theorems \ref{t:mainres1} and \ref{t:mainres4} can be solved by employing standard matrix inequalities solvers (e.g., LMI solver toolbox and PENBMI in MATLAB). 
        \end{remark}
             
        We now present numerical examples to demonstrate Theorems \ref{t:mainres1} and \ref{t:mainres4}.   
\section{Numerical example}
\label{s:numex}
    \begin{example}
    \label{ex:numex1}
     We perform the numerical experiment \cite[Experiment 1]{abc} in the presence of data losses. We consider an NCS with number of plants, \(N = 2\) and capacity of the shared communication network, \(M = 1\). Let the probability of data loss, \(q = 0.3\). 
     
     Plant \(i=1\) is a discretized version of a linearized batch reactor system presented in \cite[\S IVA]{Walsh2002} with sampling time \(0.05\) units of time. We have
    \begin{align*}
        A_1 &= \pmat{1.0795 & -0.0045 & 0.2896 & -0.2367\\-0.0272 & 0.8101 & -0.0032 & 0.0323\\
        0.0447 & 0.1886 & 0.7317 & 0.2354\\0.0010 & 0.1888 & 0.0545 & 0.9115},\\
        B_1 &= \pmat{0.0006 & -0.0239\\0.2567 & 0.0002\\0.0837 & -0.1346\\0.0837 & -0.0046}.
    \end{align*}
    Plant \(i=2\) is a discretized version of a linearized inverted pendulum system presented in \cite[\S 4]{Rehbinder2004} with sampling time \(0.05\) units of time. We have
    \begin{align*}
        A_2 = \pmat{1.0123 & 0.0502\\0.4920 & 1.0123},\:\:&\:\:B_2 = \pmat{0.0123\\0.4920}.
    \end{align*}
    
    Let \(s_1 = \{1\}\), \(s_2 = \{2\}\), \(p_1 = 0.6\) and \(p_2 = 0.4\). Clearly, \(\abs{s_1} = \abs{s_2} = 1 = M\), \(s_1\cap s_2 = \emptyset\), \(s_1\cup s_2 = \{1,2\}\) and \(p_1+p_2 = 1\). 
    
    First, we design state-feedback controllers, \(K_i\), \(i=1,2\) so that condition \eqref{e:maincondn1} is satisfied. We employ the following steps to perform this task:
    \begin{itemize}[label = \(\circ\),leftmargin = *]
        \item Plant 1:
        \begin{itemize}[label = \(\diamond\), leftmargin = *]
            \item We solve \eqref{e:auxcondn2} and obtain 
            \begin{align*}
                P_{1_s} &= \pmat{280.4127 & -2.6153 & 73.7255 & 24.1771\\
                -2.6153 & 352.0158 & 81.6804 & -15.6265\\
                73.7255 & 81.6804 & 145.2213 & -53.0460\\
                24.1771 & -15.6265 & -53.0460 & 581.9139},
                \intertext{and}
                P_{1_u} &= \pmat{638.0848 & 3.1240 & 374.2750 & -71.3485\\
                3.1240 & 582.2429 & 175.6108 & 193.2678\\
                374.2750 & 175.6108 & 357.4983 & -44.0733\\
                -71.3485 & 193.2678 & -44.0733 & 687.8599}. 
            \end{align*}
            \item We solve \eqref{e:auxcondn1} and obtain
            \begin{align*}
                Y_1 = \pmat{0.0013 & -0.0054 & 0.0010 & -0.0040\\0.0141 & -0.0023 & 0.0121 & -0.0064}.
            \end{align*}
            \item We compute \(K_1\) using \eqref{e:maincondn4} and obtain
            \begin{align*}
                K_1 = \pmat{0.9913 & -2.9132 & 0.0462 & -3.2634\\10.0804 & -0.1278 & 6.7379 & -5.3809}.
            \end{align*}
        \end{itemize}
        \item Plant 2:
         \begin{itemize}[label = \(\diamond\), leftmargin = *]
            \item We solve \eqref{e:auxcondn2} and obtain 
            \begin{align*}
                P_{2_s} &= \pmat{46.7052 & 1.1377\\1.1377 & 1.0283},
                \intertext{and}
                P_{2_u} &= \pmat{589.7951 & 180.4922\\180.4922 & 59.6479}. 
            \end{align*}
            \item We solve \eqref{e:auxcondn1} and obtain
            \begin{align*}
                Y_2 = \pmat{-0.0175 & 0.0016}.
            \end{align*}
            \item We compute \(K_1\) using \eqref{e:maincondn4} and obtain
            \begin{align*}
                K_2 = \pmat{-7.4787 -2.2188}.
            \end{align*}
        \end{itemize}
    \end{itemize} 
    
    It follows that
    \begin{align*}
        &A_{1_s}^\top \P^1 A_{1_s} - P_{1_s}\\
        =&\: \pmat{-88.9058 & -1.8199 & -42.9043 & -3.6966\\
        -1.8199 & -349.3428 & -83.4172 & -16.4636\\
        -42.9043 & -83.4172 & -139.8538 & 68.9249\\
        -3.6966 & -16.4636 & 68.9249 & -190.8335}\\
        \prec&\: 0_{4\times 4}, 
    \end{align*}
    \begin{align*}
        &A_{1_u}^\top \P^1 A_{1_u} - P_{1_u}\\
        =&\: \pmat{-45.1850 & 34.2637 & -18.6622 & -25.8573\\
        35.2637 & -161.0010 & -44.5344 & 27.6514\\
        -18.6622 & -44.5344 & -71.3559 & 25.3495\\
        -25.8573 & 27.6514 & 25.3495 & -137.0475}\\
        \prec&\: 0_{4\times 4},
    \end{align*}
    and
    \begin{align*}
        A_{2_s}^\top \P^2 A_{2_s} - P_{2_s}
        =\pmat{-1 & 0\\0 & -1}
        \prec 0_{2\times 2}, 
    \end{align*}
    \begin{align*}
        A_{2_u}^\top \P^2 A_{2_u} - P_{2_u}
        =\pmat{-1 & 0\\0 & -1}
        \prec 0_{2\times 2}.
    \end{align*}
    Thus, condition \eqref{e:maincondn1} holds for each \(i\in s_j\), \(j=1,2\).
    
    Second, we employ Algorithm \ref{algo:sched_design} and generate \(10\) different scheduling logics, \(\gamma\). We use the following numerical method to generate each logic:
    \begin{itemize}[label = \(\circ\), leftmargin = *]
        \item We fix \(T= 1000\) and set the frequency of occurrence of \(s_j\), \(j=1,2\) as \(f_{s_1} = p_1\times T = 600\) and \(f_{s_2} = p_2\times T = 400\).
        \item We construct a set \(D_1\) that contains \(f_{s_j}\)-many instances of \(s_j\), \(j=1,2\).
        \item At each time \(t=0,1,\ldots,T-1\), we pick an element \(r\) from \(D_1\) uniformly at random, assign it to \(\gamma(t)\), and update \(D_1\) to \(D_1\setminus\{r\}\). 
    \end{itemize} 
    
    Third, we generate \(10\) different data loss sequences, \(\kappa_1\). The following procedure is employed to generate each of them:
      \begin{itemize}[label = \(\circ\), leftmargin = *]
        \item We fix \(T=1000\) and set the frequency of occurrence of \(1's\) as \(f = q\times T = 300\).
        \item We construct a set \(D_2\) that contains \(f\)-many instances of \(1\)'s and the remaining elements are \(0\).
        \item At each time \(t=0,1,\ldots,T-1\), we pick an element \(\overline{r}\) from \(D_2\) uniformly at random, assign it to \(\kappa_1(t)\), and update \(D_2\) to \(D_2\setminus\{\overline r\}\).
      \end{itemize}
      
    Fourth, we pair a scheduling logic and a data loss sequence from the generated ones uniformly at random and obtain \(10\) such combinations. Corresponding to each combination, we pick \(10\) different initial conditions, \(x^{i}_0\in[-1,+1]^{d_i}\), \(i=1,2\) uniformly at random and plot \(\norm{x_i(t)}^2\), \(i=1,2\). The trajectories up to time \(t = 40\) are illustrated in Figures \ref{fig:plot1} and \ref{fig:plot2}. Stochastic stability of each plant in the NCS under consideration follows. 
    
    \begin{figure}[htbp]
    \centering  
        \includegraphics[height=6cm,width=8cm]{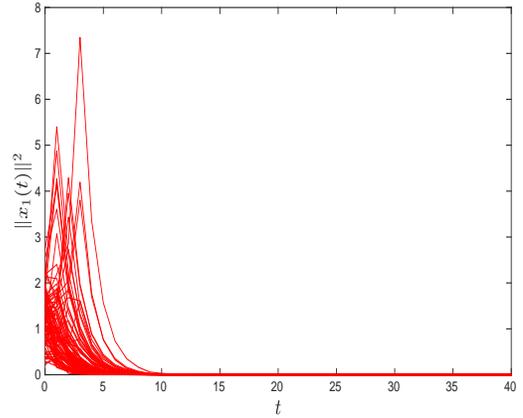}
        \caption{\(\bigl(\norm{x_1(t)}^2\bigr)_{t\geq 0}\) for Plant \(1\) in Example \ref{ex:numex1}}\label{fig:plot1}
    \end{figure}

     \begin{figure}[htbp]
    \centering  
        \includegraphics[height=6cm,width=8cm]{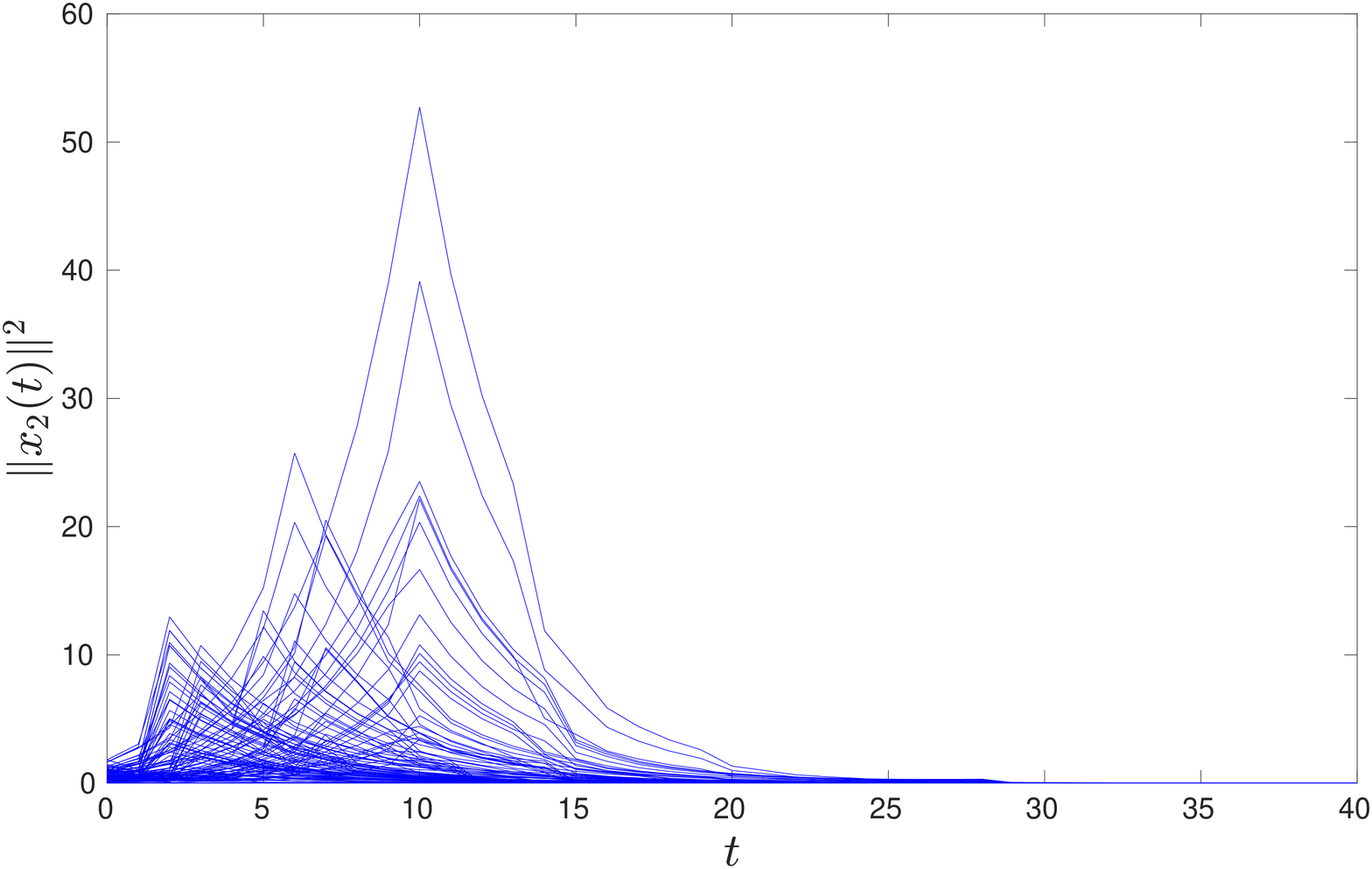}
        \caption{\(\bigl(\norm{x_2(t)}^2\bigr)_{t\geq 0}\) for Plant \(2\) in Example \ref{ex:numex1}}\label{fig:plot2}
    \end{figure}
    \end{example}
    
    \begin{example}
    \label{ex:numex2}
        We consider an NCS with number of plants, \(N=5\) and capacity of the shared communication network, \(M = 2\). Notice that \(N\% M \neq 0\) and the assumptions for the stability results \cite{abc} are not satisfied. However, the results of the current paper are applicable. Let the probability of data loss, \(q = 0.4\). 
        
        We take five copies of the same plant. Let
        \begin{align*}
            A_i = \pmat{1.2571 & -1.0259\\1.7171 & -0.6001},\:\:B_i = \pmat{1\\0},\:i=1,2,3,4,5.
        \end{align*}
        Let \(s_1 = \{1,2\}\), \(s_2 = \{3,4\}\), \(s_3 = \{5\}\), \(p_1 = 0.3\), \(p_2 = 0.3\) and \(p_3 = 0.4\). Clearly, \(\abs{s_1}=\abs{s_2}= M\), \(\abs{s_3}<M\), \(s_1\cap s_2 = \emptyset\), \(s_1\cap s_3 = \emptyset\), \(s_2\cap s_3 = \emptyset\), \(s_1\cup s_2\cup s_3=\{1,2,3,4,5\}\) and \(p_1+p_2+p3=1\).
        
        First, we design state-feedback controllers, \(K_i\), \(i=1,2,3,4,5\) so that condition \eqref{e:maincondn1} is satisfied. We employ the following steps to perform this task:
        \begin{itemize}[label=\(\circ\),leftmargin=*]
            \item We solve \eqref{e:auxcondn2} and obtain
            \begin{align*}
                P_{i_s} &= \pmat{764.8674 & -280.5277\\-280.5277 & 140.1558},\:i=1,2,3,4,5,\\
                \intertext{and}
                P_{i_u} &= \pmat{588.7880 & -339.0439\\-339.0439 & 388.3214},\:i=1,2,3,4,5.
            \end{align*}
            \item We solve \eqref{e:auxcondn1} and obtain
            \begin{align*}
                Y_i = \pmat{0.0044 & 0.0153},\:i=1,2,3,4,5.
            \end{align*}
            \item We compute \(K_i\), \(i=1,2,3,4,5\) using \eqref{e:maincondn4} and obtain
            \begin{align*}
                K_i = \pmat{-0.8994 & 0.9005},\:i=1,2,3,4,5.
            \end{align*}
        \end{itemize}
        It follows that for \(i=1,2,3,4\),
        \begin{align*}
            A_{i_s}^\top \P^i A_{i_s} - P_{i_s} &= \pmat{-75.7909 & 39.8393\\39.8393 & -56.0854}\prec 0_{2\times 2},\\
            A_{i_u}^\top \P^i A_{i_u} - P_{i_u} &= \pmat{-13.2316 & 11.2439\\11.2439 & -16.0168}\prec 0_{2\times 2},
        \end{align*}
        and for \(i=5\),
         \begin{align*}
            A_{i_s}^\top \P^i A_{i_s} - P_{i_s} &= \pmat{-114.0282 & 53.1988\\53.1988 & -60.7530}\prec 0_{2\times 2},\\
            A_{i_u}^\top \P^i A_{i_u} - P_{i_u} &= \pmat{-25.2808 & 4.1286\\4.1286 & -5.9368}\prec 0_{2\times 2}.
        \end{align*}
        Thus, condition \eqref{e:maincondn1} holds for each \(i\in s_j\), \(j=1,2,3\).
        
        Second, we employ Algorithm \ref{algo:sched_design} and generate \(10\) different scheduling logics, \(\gamma\). The numerical method described in Example \ref{ex:numex1} is used for the purpose of sequence generation.
        
        Third, corresponding to each \(\gamma\) generated above, we generate data loss sequences, \(\kappa_m\), \(m=1,2\). In particular, we keep \(\kappa_1 = \kappa_2\). We employ the sequence generation procedure described in Example \ref{ex:numex1} for this purpose.
        
        Fourth, corresponding to each combination of \(\gamma\) and \(\kappa_m\), \(m=1,2\), we pick \(10\) different initial conditions, \(x_i^0\in[-1,+1]^2\), \(i=1,2,3,4,5\) uniformly at random and plot the resulting trajectory \(\norm{x_i(t)}^2\), \(i=1,2,3,4,5\). In Figures \ref{fig:plot3} and \ref{fig:plot4} we illustrate one scheduling logic \(\gamma\), one data loss sequence \(\kappa_m\), \(m=1,2\), respectively. In Figure \ref{fig:plot5} we illustrate the trajectories of the plant obtained from our experiment up to time \(t = 100\). Stochastic stability of each plant in the NCS under consideration follows.
         \begin{figure}[htbp]
    \centering  
        \includegraphics[height=6cm,width=8cm]{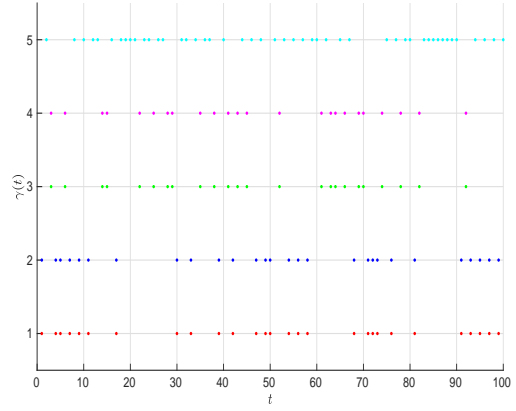}
        \caption{One choice of \(\bigl(\gamma(t)\bigr)_{t\geq 0}\) for Example \ref{ex:numex2}}\label{fig:plot3}
    \end{figure}
     \begin{figure}[htbp]
    \centering  
        \includegraphics[height=6cm,width=8cm]{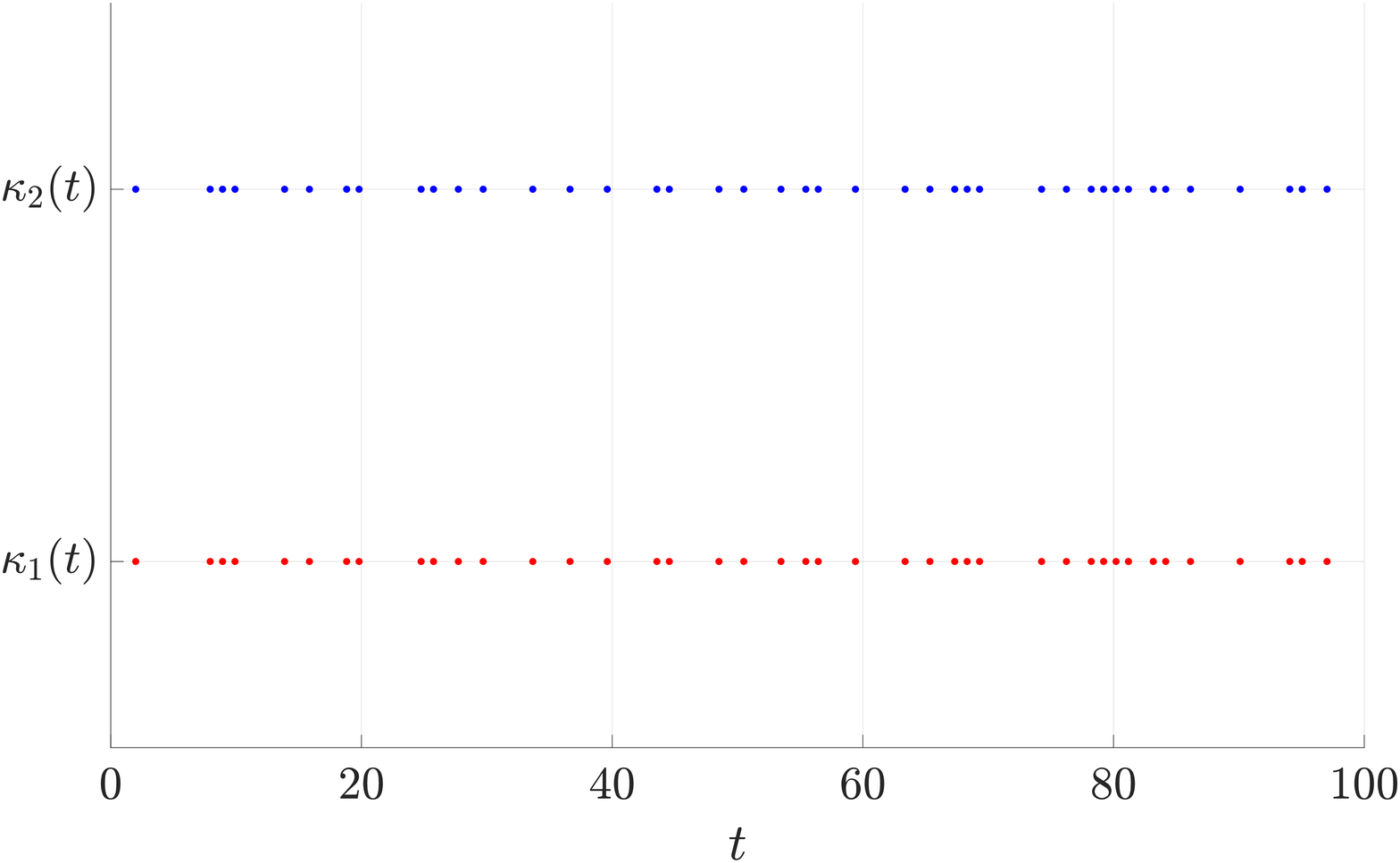}
        \caption{One choice of \(\bigl(\kappa_m(t)\bigr)_{t\geq 0}\), \(m=1,2\) for Example \ref{ex:numex2}}\label{fig:plot4}
    \end{figure}
    \begin{figure}[htbp]
    \centering  
        \includegraphics[height=6cm,width=9cm]{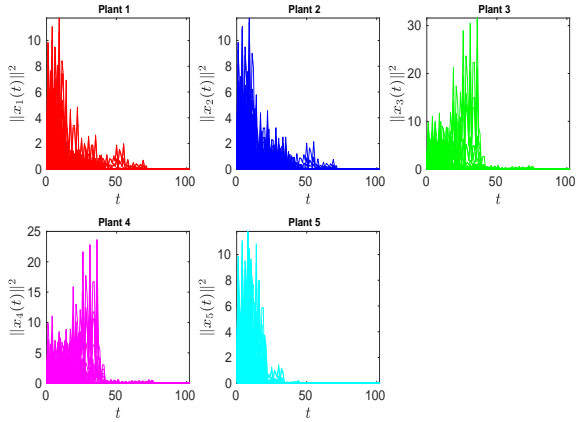}
        \caption{\(\bigl(\norm{x_i(t)}^2\bigr)_{t\geq 0}\), \(i=1,2,3,4,5\) for Example \ref{ex:numex2}}\label{fig:plot5}
    \end{figure}
    \end{example}

\section{Concluding remarks}
\label{s:concln}
    In this paper we discussed a probabilistic algorithm that designs scheduling logics for NCSs whose shared communication networks have a limited communication capacity and are prone to data losses. A scheduling logic obtained from our algorithm preserves stochastic stability of each plant in the NCS. We assumed that the data loss pattern follows a probabilistic model with a known probability. In practice, however, a precise estimate of the data loss probability is often not available. Further, we considered the data loss probability to be the same in all the channels in the network. This assumption is restrictive and may be generalized to different data loss probabilities for different channels in the network. We identify a generalization of our results to the setting where the data loss pattern follows a probabilistic model with different and unknown probabilities in various channels in the network, as a future research direction.

\end{document}